\documentclass{amsart}
\usepackage[utf8]{inputenc}

\usepackage{amsmath}
\usepackage{amsthm}
\usepackage{amssymb}
\usepackage{color}
\usepackage{todonotes}
\usepackage{stmaryrd}

\newtheorem{theorem}{Theorem}
\newtheorem{proposition}[theorem]{Proposition}
\newtheorem{corollary}[theorem]{Corollary}
\newtheorem{lemma}[theorem]{Lemma}

\theoremstyle{definition}
\newtheorem{definition}[theorem]{Definition}

\newtheorem{example}{Example}

\newcommand{\F}{{\mathbb{F}}}
\newcommand{\R}{{\mathbb{R}}}
\newcommand{\N}{{\mathbb{N}}}
\newcommand{\floor}[1]{\left\lfloor#1\right\rfloor}
\newcommand{\polyring}[1]{\F_2[x_1,\dots,x_m]^{#1}}
\newcommand{\ev}{\textnormal{ev}_{\F_2^m}}

\newcommand{\short}{\textnormal{short}}
\newcommand{\punct}{\textnormal{punct}}
\newcommand{\supp}{\textnormal{supp}}
\newcommand{\CSS}{\textnormal{CSS}}

\DeclareMathOperator{\RM}{RM}

\title{CSS-T Codes from Reed-Muller Codes} 
\author[]{Emma Andrade, Jessalyn Bolkema, Thomas Dexter, Harrison Eggers, Victoria Luongo, Felice Manganiello, Luke Szramowski}
\date{}

\begin{document}

\maketitle

\begin{abstract}
    CSS-T codes are a class of stabilizer codes introduced by Rengaswamy \emph{et al} with desired properties for quantum fault-tolerance. In this work, we comprehensively study non-degenerate CSS-T codes built from Reed-Muller codes. These classical codes allow for constructing CSS-T code families with nonvanishing asymptotic rates up to $\frac{1}2$ and possibly diverging minimum distance when non-degenerate. 
\end{abstract}

\section{Introduction}

Quantum error correcting codes (QECC) and methods to achieve fault-tolerance are essential to overcome the large number of errors that happen to data and computations in a quantum computer. A set of quantum gates that are universal is $\{H,S,CX,T\}$, where the Hadamard gate $H$, the phase gate $S$, and the controlled-NOT gate $CX$ form a generating set of the Clifford group and $T$ is a non-Clifford gate. Given a $\llbracket n,k,d \rrbracket$ stabilizer QECC, a challenging problem is to characterize the gates on the $k$ logical qubits encoded using transversal gates on the $n$ physical qubits. Transversal gates are desirable for fault-tolerance as they prevent the spread of errors between qubit levels.

Rengaswamy \emph{et al.} provides in \cite{CSS-T-it,CSS-T-isit}   algebraic conditions on stabilizer QECC for their codespace to be preserved under physical transversal $T$ and $T^\dagger$ gates. This is equivalent to saying that these physical gates encode some logical gates.  
 They also prove that, for any non-degenerate stabilizer code that admits transversal $T$ and $T^\dagger$, a CSS code with the same parameters admitting the same gates exists. This implies that the latter family of codes, denoted CSS-T, is an optimal family of stabilizer codes admitting transversal $T$ and $T^\dagger$. Note that this is a unique property of CSS-T codes since it is common knowledge that CSS codes are generally not optimal in the family of stabilizer codes. CSS codes are a well-studied family of stabilizer codes that rely on classical codes for their definition. Many families of CSS codes have been constructed since their definition in  \cite{CSofCSS} and \cite{SofCSS}; for an up-to-date list of CSS code constructions, please visit  \cite{eczoo_css}.  

\newpage
In \cite{CSS-T-isit}, the authors propose an open problem: find families of CSS-T codes with nonvanishing asymptotic rate and asymptotic relative distance. These properties would be desirable as they are a prerequisite for constant overhead magic state distillation (MSD) \cite{MSD}. 
Motivated by this challenge, we comprehensively study CSS-T codes from Reed-Muller codes in this work. Other constructions of quantum error correcting codes based on Reed-Muller codes have been considered in \cite{zh97,st99,CSS-T-it}, but never considered for CSS-T code construction. These codes allow constructing CSS-T code families with nonvanishing asymptotic rates up to $\frac{1}2$. Although the family does not have a nonvanishing relative minimum distance if the codes are non-degenerate,  the minimum distance diverges in specific cases, 
giving a partial answer to the challenge set out in \cite{CSS-T-isit}. Our initial findings were presented as a preprint \cite{ourpaper}, and subsequent studies, such as \cite{berardini_structure_2024,camps-moreno_algebraic_2024,campsmoreno2024quantumcsstcodessparse,csstallerton}, have provided validations of this work. Notably, \cite[Theorem 4.9]{berardini_structure_2024} demonstrates that the family of CSS-T codes examined in this manuscript is optimal with respect to both relative minimum distance and code rate.

This paper is organized as follows. In Section \ref{s:def+not}, we establish the notation we use throughout the manuscript. We recall the definition of CSS-T codes and state the problem we address. 
In Section \ref{s:asympt}, we study the asymptotic properties of families of Reed-Muller codes, focusing on the asymptotic rate and relative distance. In this section, we establish fundamental limitations of working only with Reed-Muller codes, meaning that families of CSS codes defined only by Reed-Muller codes must have vanishing asymptotic relative distance. In Section \ref{s:CSS-T-RM}, we focus on constructing CSS-T codes specifically using Reed-Muller codes. We translate the defining properties of Reed-Muller codes to the classical operations of puncturing and shortening linear codes and prove the main theorem, Theorem \ref{t:css-t-family}. We conclude in Section \ref{s:examples} with some clarifying examples. 

\section{Definitions and Notations}\label{s:def+not}

This section states the definitions and the notation needed throughout the manuscript. For definitions of stabilizer codes and other basic properties, we refer the reader to \cite{lidar_brun_2013,gottesman1997stabilizer}, whereas we refer to \cite{huffman_pless_2003} for a reader interested in classical coding theory.

The Calderbank-Shor-Steane (CSS) codes \cite{CSofCSS,SofCSS} are a special class of stabilizer codes constructed from classical error-correcting codes. Specifically,  if $C_1$ is  $[n,k_1,d_1]$, $C_2$ is  $[n,k_2,d_2]$, $C_2\subseteq C_1$, and $d_2^\perp$ denotes the minimum distance of $C_2^\perp$, then $CSS(C_1,C_2)$ is a $\llbracket n, k_1-k_2, \geq \min\{d_1,d_2^\perp\}\rrbracket$ quantum stabilizer code. In this paper, we consider the case of non-degenerate CSS codes. If $G_2$ and $G_1^\perp$ are generator matrices for the codes $C_2$ and $C_1^\perp$, respectively, then one can write a generator matrix for the binary representation of stabilizers as 
\[G_S = \left[\begin{array}{c|c} 0  & G_1^\perp \\ \hline G_2 & 0\end{array}\right].\]

Rengaswamy \emph{et al.} suggest in \cite{CSS-T-it,CSS-T-isit} that a subfamily of CSS codes supports a physical transversal $T$ and $T^\dagger$ gates. This is equivalent to saying that these physical gates encode some logical gates. The subfamily is defined as follows.

\begin{definition}[CSS-T codes \cite{CSS-T-isit}] \label{d:csst}
A CSS-T code is a CSS code associated with a pair $(C_1, C_2)$ of binary linear codes satisfying 
\begin{enumerate} 
\item $C_2$ is an even code, meaning that all codewords of $C_2$ have even weight, and
\item for each codeword $x \in C_2$, there exists a self-dual code in $C_1^\perp$ that is supported on $x$. 
\end{enumerate}
\end{definition}

We elaborate on this definition and its equivalent formulations in Section \ref{s:CSS-T-RM}. 

Throughout, we consider Reed-Muller codes as candidates for both $C_1$ and $C_2$ in the definition of a CSS-T code. 
We denote by $\polyring{r}$ the set of all polynomials in $m$ variables with a total degree less than or equal to $r$. For $p\in \polyring{}$, we denote by $\ev(p)$ the vector obtained by evaluation of $p$ at $\F_2^m$; hence if $\F_2^m=\{v_1,\dots,v_{2^m}\}$, then 
\[\ev(p)=(p(v_1),\dots,p(v_{2^m})).\]

\begin{definition}
Given $r, m \in \N$ with $r\leq m$, the $r^{th}$ order Reed-Muller code of length $2^m$, denoted $\RM(r,m)$, is the subspace of $\F_2^{2^m}$ give by 
\[\RM(r,m) = \{ \ev(f) \mid f \in \F_2[x_1, x_2, \ldots, x_m]^r\}.\]
\end{definition}

In the following sections, we utilize the fact that $\RM(r,m)$ is a $[2^m, \sum_{i=0}^r {m\choose i}, 2^{m-r}]$ linear code, as well as the nested property of Reed-Muller codes, meaning that \[\RM(r_2, m) \subseteq \RM(r_1,m)\] if $r_2 \leq r_1$. Finally, we recall that $\RM(r,m)^\perp = \RM(m-r-1, m)$. 

If $C$ is a linear code of length $n$ and $x\in C$, then $\supp(x)=\{i\in [n] \mid x_i\neq 0\}$. If $C=\RM(r,m)$ is a Reed-Muller code, then $\F_2^m$ is its index set, and if $x\in C$, then $\supp(x)\subseteq \F_2^m$. If $p\in \polyring{r}$, then we denote with $\supp(p)$ the support of $x=\ev(p)\in \RM(r,m)$, meaning that $\supp(p)=\supp(x)$.

\section{Reed-Muller Codes' Asymptotic Properties }\label{s:asympt}

We begin our analysis with the fundamental parameters of Reed-Muller codes, establishing asymptotic relative minimum distance in Proposition~\ref{prop:dist} and asymptotic rate in Proposition~\ref{prop:rate}.
\begin{proposition}
Let $C(s)= \RM(r(s), m(s))$ for $s\in \N$ with $\lim_{s\to \infty}m(s)=\infty$. The sequence of codes $(C(s))_{s\in \N}$ has a nonvanishing relative minimum distance if and only if $\lim_{s \to \infty} r(s) = c$.
\label{prop:dist}
\end{proposition}

\begin{proof} 
The code $\RM(r,m)$ has minimum distance $2^{m-r}$ and therefore relative minimum distance $2^{m-r}/2^m=2^{-r}$. Thus, the asymptotic relative minimum distance of $C(s)$ is nonzero exactly when $\lim_{s\to \infty} r(s)=c$ for some $c\in \R$.
\end{proof}

\begin{proposition}
Let $C(s)=\RM(r(s),m(s))$ for $s\in \N$ such that $\lim_{s\to \infty}m(s)=\infty$. Let $R(C(s))$ be the rate of code $C(s)$. The asymptotic rate of the sequence $(C(s))_{s\in \N}$ is     \[\lim_{s\to \infty}R(C(s)) = \Phi\left (\lim_{s \to \infty} \left (\frac{2r(s) - m(s)}{\sqrt{m(s)}}\right)\right)\] where $\Phi$ is the standard normal Gaussian CDF.
\label{prop:rate}
\end{proposition}

Such analysis of the asymptotic rate of Reed-Muller codes appears elsewhere in the literature; see \cite{pfister-friends}. We include a proof in the interest of completeness. 

\begin{proof}
Consider a binomial random variable $X\sim B(m,\frac{1}{2})$. The rate of the code $\RM(r,m)$ is given by \[ \sum_{i=0}^r {m \choose i}\left(\frac{1}{2}\right)^m \] which coincides exactly with the cumulative probability $P(X\leq r)$. As $m\to\infty$, by the Central Limit Theorem, $X$ tends towards a normal distribution $\mathcal{N}(m/2,m/4)$ with associated cumulative density function given by  $P(X\leq r)=\Phi\left(\frac{2r-m}{\sqrt{m}}\right)$. Thus as $m\to \infty$, the rate of the code $\RM(r,m)$ behaves as $\Phi\left(\frac{2r-m}{\sqrt{m}}\right)$. That is, 
\begin{align*}
    \lim_{s\to \infty} R(C(s))   & = \lim_{s\to \infty} \Phi \left (\frac{2r(s) - m(s)}{\sqrt{m(s)}}\right)\\ 
                                & =  \Phi\left (\lim_{s \to \infty} \left (\frac{2r(s) - m(s)}{\sqrt{m(s)}}\right)\right)
\end{align*}
\end{proof}

We can now establish sufficient conditions for the nonvanishing rate. 
\begin{proposition}
Let $C(m)=\RM(\floor{\frac{m-1}{2}} - t(m), m)$ for $m\in \N$. Then $(C(m))_{m\in \N}$ has nonvanishing asymptotic rate if and only if\[\lim_{m\to \infty} \frac{t(m)}{\sqrt{m}}\neq \infty.\]
\label{prop:vanish}
More specifically, it holds that 
\begin{itemize}
\item $\mathcal{R}\in [0,\frac{1}{2}]$ if $\lim_{m\to \infty} \frac{t(m)}{\sqrt{m}}\geq 0$, and 
\item $\mathcal{R}\in [\frac{1}{2},1]$ if  $\lim_{m\to \infty} \frac{t(m)}{\sqrt{m}}\leq 0$.
\end{itemize}
\end{proposition}
\begin{proof}
By previous proposition, the family ${C(m)}_{m\in \N}$ has asymptotic rate \[\mathcal{R}=\Phi\left(\lim_{m\to\infty} \frac{2 \floor{\frac{m-1}{2}}- 2t(m) -m}{\sqrt{m}} \right)=\Phi\left(- 2 \lim_{m\to\infty} \frac{ t(m)}{\sqrt{m}} \right).  \] Thus $\mathcal{R}\neq 0$ if and only if \[\Phi\left(\lim_{m\to\infty} \frac{2 \floor{\frac{m-1}{2}}- 2t(m) -m}{\sqrt{m}} \right) \neq 0,\] or equivalently, \[\lim_{m\to\infty} \frac{2 \floor{\frac{m-1}{2}}- 2t(m) -m}{\sqrt{m}} = - 2 \lim_{m\to\infty} \frac{ t(m)}{\sqrt{m}} \neq -\infty,\]  which corresponds to 
\[\lim_{m\to \infty} \frac{t(m)}{\sqrt{m}}= \infty.\]

\end{proof}

Alas, Proposition~\ref{prop:dist} and Proposition~\ref{prop:rate} together mean that a family of Reed-Muller codes cannot have both nonvanishing rate and nonvanishing minimum relative distance.
\begin{proposition}
A family of Reed-Muller codes $\RM(r(m), m)$ cannot have both nonvanishing rate
and nonvanishing relative distance.

\end{proposition}

\begin{proof}
Suppose the family $\{\RM(r(m),m)\}_{m\in \N}$ has nonvanishing relative distance, and so $r(m)\to c$ for some $c\in \R$ as $m\to\infty$. It follows that \[\lim_{m \to \infty} \left (\frac{2r(m) - m}{\sqrt{m}}\right) = \lim_{m\to \infty} \left( \frac{2c-m}{\sqrt{m}} \right)= -\infty.\] The asymptotic rate of $\{\RM(r(m),m)\}_{m\in \N}$ can thus be evaluated 
\[ \mathcal{R} = \Phi(-\infty) = 0. \]
\end{proof}

\section{CSS-T codes from Reed-Muller Codes}\label{s:CSS-T-RM}

Having established the asymptotic parameters of classical Reed-Muller codes, we are now prepared to analyze the corresponding parameters in the quantum case. First, we observe that a CSS code constructed from RM codes must necessarily have a vanishing (quantum) distance.

    

\begin{corollary}
Let $C_1(s)= RM(r_1(s), m(s))$ and $C_2(s)= RM(r_2(s), m(s))$ for $s\in \N$ with $\lim_{s\to \infty}m(s)=\infty$. The sequence of CSS codes \[(\CSS(C_1(s),C_2(s)))_{s\in \N},\] if non-degenerate, has vanishing relative minimum distance.
\end{corollary} 

\begin{proof}
    Recall that $\CSS(C_1,C_2)$ requires $C_2\subseteq C_1$, so $r_2(s) \leq r_1(s)$ for all $s\in\N$. Suppose, seeking contradiction, that the sequence of CSS codes has a nonvanishing relative minimum distance. Then, we have $r_1(s)\to c_1$, $r_2^\perp(s)\to c_2$ as $s\to\infty$, for some $c_1,c_2\in\R$. However, $r_2^\perp(s) = m(s) - r_2(s) - 1$, so $r_2^\perp(s)\to c_2$ and $m(s)\to\infty$ implies that $r_2(s)\to\infty$. It follows from $r_2(s) \leq r_1(s)$ that $r_1(s)\to\infty$, a contradiction. So, such a sequence of CSS codes has always vanishing relative minimum distance.
\end{proof}




Now, we focus on characterizing all CSS-T codes, specifically those defined by Reed-Muller codes. Recall, from Definition~\ref{d:csst}, that the fundamental property that makes a CSS code a CSS-T code involves self-dual codes contained in restrictions of dual codes. This leads to the following lemma. 

\begin{lemma}\label{l:rate} 

 The maximum rate of a $CSS(C_1,C_2)$ CSS-T code where  $C_2$ is a RM code is $\frac{1}2$.
 \end{lemma}

\begin{proof}
    Since $C_2$ is a RM code, it holds that $(1,\dots,1)\in C_2$. Because $\supp(1,\dots,1)=[n]$, by Definition~\ref{d:csst} there exists a self-dual code contained in $C_1^\perp$, which implies that $\dim C_1^\perp \geq \frac{n}{2}$. The bound on the rate of the CSS-T code follows from the following bound on its dimension  
    \[\dim C_1-\dim C_2\leq \dim C_1 = n-\dim C_1^\perp \leq \frac{n}{2}.\]
\end{proof}



The following proposition helps us reformulate the definition of a CSS-T code.

\begin{proposition}\label{p:self-ort}
An $[n, k]$ binary linear code $C$ contains a self-dual code if and only if $n$ is even and $C^\perp$ is self-orthogonal, meaning that $C^\perp\subseteq C$.
\end{proposition}

\begin{proof}
  We begin by proving the forward direction. Let $C$ be an $[n,k]$ linear code that contains a self-dual code $D=D^\perp$. Since $D\subseteq C$, we see $C^\perp\subseteq D^\perp=D$. Combining these two relationships, we see $C^\perp\subseteq D\subseteq C$. Thus, $C^\perp$ is self-orthogonal.

We now prove the backward implication. Define \[S=\left\{\overline{C}\mid  C^\perp\subset \overline{C}\subset C, \dim(\overline{C})=\frac{n}{2}\right\}\] to be the set of subcodes of $C$ containing $C^\perp$ of dimension half of their length. 
It holds that $\overline{C}\in S$ if and only if $\overline{C}^\perp\in S$. Indeed, if $\overline{C}\in S$ we know that  $C^\perp\subseteq\overline{C}\subseteq C$. This implies that $C^\perp\subseteq\overline{C}^\perp\subseteq (C^\perp)^\perp=C$ and $\dim(\overline{C}^\perp)=\frac{n}{2}$, so $\overline{C}^\perp\in S$. 

 We use the Gaussian binomial coefficient to calculate the cardinality of $S$. To account only for those codes that contain $C^\perp$, we consider the chain quotient spaces 
 \[C^\perp/C^\perp=\{0\}\subseteq\overline{C}/C^\perp\subseteq C/C^\perp.\] Note that for every $\overline{C}\in S$, there is exactly one $\overline{C}/C^\perp\subseteq C/C^\perp$, so by counting the number of subspaces of the form $\overline{C}/C^\perp$ that are contained in $C/C^\perp$, we find the cardinality of $S$. The dimension of $\overline{C}/C^\perp$ is \[\dim \overline{C}/C^\perp=\dim(\overline{C})-\dim(C^\perp)=\frac{n}{2}-(n-k)=k-\frac{n}{2},\] and the dimension of $C/C^\perp$ is $\dim(C)-\dim(C^\perp)=k-(n-k)=2k-n$. 
 
 Since we know the dimension of each of these spaces, we can now use the Gaussian binomial coefficient to compute the cardinality of $S$ as follows:
\[
|S| = \left[\begin{smallmatrix} 2k-n \\ k-\frac{n}{2} \end{smallmatrix}\right]_2=\prod_{i=0}^{k-\frac{n}{2}-1}\frac{2^{2k-n-i}-1}{2^{i+1}-1}.
\]
Note that $2^{2k-n-i}-1$ and $2^{i+1}-1$ are both odd for all $i\in\mathbb{N}$. Therefore, $|S|$ is odd. 

Recall that each code $\overline{C}$ in $S$ is paired off with exactly one other code, its dual. Since $S$ is odd, at least one self-dual code must be in $S$. Therefore, a self-dual subcode exists in $C$.
\end{proof}

We note that the cardinality of the set $S$ can alternatively be deduced from \cite[Proposition 2.5]{eamer-alberto}.

\begin{definition} 
Given a code $C\subseteq \F_2^n$ and index set $I\subseteq \{1,\dots, n\}$, produce the code $C$ punctured on $I$, denoted $\punct(C,I)$, by removing from every $x\in C$ the coordinates in $I$. The code $C$ shortened on $I$, denoted $\short(C,I)$, consists of $x\in C$ such that $x_i=0$ for all $i\in I$ with coordinates in $I$ removed. 
\end{definition}

The following classical result relating punctured and shortened codes is found in Theorem 1.5.7 of \cite{huffman_pless_2003}.
\begin{theorem}\label{t:short-punct}
Let $C$ be a $[n,k]$ binary linear code and $I\subseteq \{1,\dots,n\}$, then \[\short(C,I)^\perp=\punct(C^\perp,I)\]
\end{theorem}

It follows that one can express the defining property of CSS-T codes in terms of punctured and shortened codes.

\begin{corollary}
Let $C_1$ and $C_2$ be two linear codes and $C=\CSS(C_1,C_2)$ a CSS code. $C$ is a CSS-T code if and only if for all $x\in C_2$ it holds that \[\punct(C_1,\F_2^m\setminus \supp(x))\subseteq \short(C_1^\perp,\F_2^m\setminus \supp(x)).\]
\label{cor:short}
\end{corollary}

\begin{proof}
We only need to focus on the self-dual containment condition. Let $x\in C_2$, so  by Proposition \ref{p:self-ort} we have that $C_1^\perp|_{\supp(x)}$ contains a self-dual code if and only if $(C_1^\perp|_{\supp(x)})^\perp$ is self-orthogonal, meaning that \[(C_1^\perp|_{\supp(x)})^\perp\subseteq C_1^\perp|_{\supp(x)}.\] Since $C_1^\perp|_{\supp(x)}  =\short(C_1^\perp,\F_2^m\setminus {\supp(x)})$ and thanks to Theorem \ref{t:short-punct}, the condition translates into
\begin{align*}
    \punct(C_1,\F_2^m\setminus {\supp(x)})&=\short(C_1^\perp,\F_2^m\setminus {\supp(x)})^\perp\\ &\subseteq \short(C_1^\perp,\F_2^m\setminus {\supp(x)}).
\end{align*}

\end{proof}

We can now completely characterize CSS-T codes from Reed-Muller codes. 
\begin{theorem}\label{t:css-t}
Let $C_1=RM(\floor{\frac{m-1}{2}} - t, m)$ and  $C_2=RM(r_2,m)$. Then $\CSS(C_1,C_2)$ is a CSS-T code if and only if $r_2\leq 2t+1$ for $m$ even, or 
$r_2\leq 2t$ for $m$ odd.
\end{theorem}

\begin{proof}
Let $x\in C_2$. Then there exists a polynomial $p\in \polyring{r_2}$ such that $x=\ev(p)$. For simplicity in notation, let $r_1=\floor{\frac{m-1}{2}} - t$, so $C_1^\perp=\RM(m-r_1-1,m)$. Let $y\in C_1$, so there exists a polynomial $q\in \polyring{r_1}$ such that $y=\ev(q)$. 

Consider the polynomials $p$ and $1-p$. Because we are working over the binary field, it holds that $\supp(1-p)$ and $\supp(p)$ partition the set $\F_2^m$, meaning that $\supp(1-p)\cup \supp(p)=\F_2^m$ and $\supp(1-p)\cap \supp(p)=\emptyset$. This gives us a tool to split $\supp(q)$; specifically, given the polynomial decomposition $q=(1-p)q+pq$, it holds that $\supp((1-p)q)\subseteq \supp(1-p)$ and $\supp(pq)\subseteq \supp(p)$. Moreover, by the linearity of the evaluation map, it follows that $\ev(q)=\ev((1-p)q)+\ev(pq)$. 

Recall that $\supp(p)=\supp(x)$ whenever $x=\ev(p)$. Then

\begin{align*}
    \punct(y,\F_2^m\setminus \supp(x))&= \punct(\ev(q),\F_2^m\setminus \supp(p))\\
    &= \punct(\ev((1-p)q)+\ev(pq),\supp(1-p))\\
    &= \punct(\ev(pq), \supp(1-p))\\&=\short(\ev(pq),\F_2^m\setminus \supp(p))
\end{align*}

where the last equality holds because $\supp(pq)\subseteq \supp(p)$.

To conclude, it is enough to prove that $pq\in\polyring{m-r_1-1}$, implying that $\ev(pq)\in C_1^\perp$. Indeed, 
\begin{align*}
    \deg(pq)&\leq r_1+r_2 = \floor{\frac{m-1}{2}} - t +r_2\\
    &=\begin{cases}
        h-1-t+2t+1, & m=2h \\ 
        h-t+2t, & m=2h+1
    \end{cases}\\
    &=h+t=m-r_1-1.
\end{align*}
Thus by Corollary~\ref{cor:short}, $\CSS(C_1,C_2)$ is indeed CSS-T.
\end{proof}

We conclude this section with the final theorem joining all of our results.

\begin{theorem}\label{t:css-t-family}  
Let $t(m)$ be such that $\lim_{m\to \infty} t(m)=\infty$. Let   \[C_1(m)=\RM\left(\floor{\frac{m-1}{2}} - t(m), m\right) \mbox{ and }  C_2(m)=\RM(r_2(m),m)\] where $r_2(m)\leq 2t(m)+1$ for $m$ even, or 
$r_2(m)\leq 2t(m)$ for $m$ odd. Then, for $m\in \N$, the code $\CSS(C_1(m),C_2(m))$ is a \[\left\llbracket n, \floor{\frac{m-1}{2}} - t(m) - r_2(m), \geq 2^{r_2(m)+1}\right\rrbracket\] CSS-T code.  Furthermore, If $\lim_{m\to \infty} \frac{t(m)}{\sqrt{m}}=c<\infty$ where $c$ is a non-negative constant, the code $\CSS(C_1(m),C_2(m))$ has non-vanishing asymptotic rate \[\Phi(-2c)\leq \frac{1}2.\] Else, if  $\lim_{m\to \infty} \frac{t(m)}{\sqrt{m}}=\infty$, the code has vanishing rate.

\end{theorem}

\begin{proof}
Theorem \ref{t:css-t} implies that  codes $\CSS(C_1(m),C_2(m))$ are CSS-T codes for any $m\in \N$.

The dimension of each CSS-T code is a direct consequence of the formula of the dimension of a CSS code based on codes $C_1(m)$ and $C_2(m)$, meaning that
\[\dim(\CSS(C_1(m),C_2(m)))=\dim C_1(m)-\dim C_2(m)=\floor{\frac{m-1}{2}} - t(m)-r_2(m).\]
For the minimum distance, for simplicity, denote with $r_1(m)$ maximum degree of the multivariate polynomials defining the codewords of $C_1(m)$. Then, by the definition of minimum distance of a CSS code, 
\begin{align*}
    d&\geq \min\{d_1, d_2^\perp\}=\min \{2^{m-r_1(m)},2^{m-(m-r_2(m)-1)}\}=2^{r_2(m)+1}
\end{align*}
since $r_2(m)\leq r_1(m)\leq \floor{ \frac{m-1}2}$ by Lemma \ref{l:rate}.

Finally, the asymptotic rate of the family of codes $(C_2(m))_{m\in \N}$ is vanishing since 
\begin{align*}
    \lim_{m \to \infty} \frac{2r_2(m) - m}{\sqrt{m}}&=\lim_{m \to \infty}\frac{2(2t(m)+1) - m}{\sqrt{m}}\\
    &= 4\lim_{m \to \infty}\frac{t(m)}{\sqrt{m}} +\lim_{m \to \infty}\frac{2}{\sqrt{m}}-\lim_{m \to \infty} \frac{m}{\sqrt{m}}=-\infty
\end{align*}
As a consequence, the asymptotic rate of the family of CSS-T codes defined by $\CSS(C_1(m), C_2(m))$ is the same as the family of Reed-Muller codes $(C_1(m))_{m\in \N}$ given by
\begin{align*}
    \Phi\left(\lim_{m \to \infty} \frac{2\left(\floor{\frac{m-1}{2}}-t(m)\right)-m)}{\sqrt{m}}\right) = \Phi\left(\lim_{m \to \infty} -2\frac{t(m)}{\sqrt{m}}\right)=\Phi(-2c)\leq \frac{1}{2}.
\end{align*}
by Proposition \ref{prop:vanish}.\end{proof}
\section{Examples}\label{s:examples}
Finally, we note that the complete characterization of Theorem~\ref{t:css-t-family} subsumes and expands upon existing examples of CSS-T codes in the literature. 

\begin{example}  
Let $t(m) = 0$ and $r_2(m) = 0$. Take $C_1(m)$ and $C_2(m)$ as defined in Theorem \ref{t:css-t-family}. Then, for all $m\in\N$, $\CSS(C_1(m), C_2(m))$ forms a CSS-T code. Further, by Theorem \ref{t:css-t-family}, the asymptotic rate of this family is 
\[ \Phi(0) = \frac{1}{2}. \]
This family of CSS-T codes achieves the maximum possible asymptotic rate.
\end{example}



\begin{example}  
In \cite{CSS-T-it}, a family of CSS-T codes from Reed-Muller codes is proposed. We describe an equivalent family of codes using the techniques from Theorem \ref{t:css-t-family}.

Let $m$ be divisible by 3. Then, define $t(m)$ as follows:
    \begin{align*}
        t(m) &= \frac{1}{3}\floor{\frac{m-1}{2}} - \frac{2}{3}   &\textrm{($m$ even)} \\
        t(m) &= \frac{1}{3}\floor{\frac{m-1}{2}} - \frac{1}{3}   &\textrm{($m$ odd)}.
    \end{align*}
    Let $r_1$ be defined as in Theorem \ref{t:css-t-family}, and let $r_2$ be equal to its upper bound as stated in Theorem \ref{t:css-t-family}. That is, 
    \begin{align*}
        r_2(m) &= 2t(m) + 1  &\textrm{($m$ even)} \\
        r_2(m) &= 2t(m) &\textrm{($m$ odd)}.
    \end{align*}
    Then, $C_1=\RM(r_1, m)$ and $C_2=\RM(r_2, m)$ define a CSS-T code by Theorem \ref{t:css-t-family}.
\end{example}
\begin{example}
Let us consider the functions of the form:
$$t(m) = \lfloor c \sqrt{m} \rfloor$$
for some positive $c \in \mathbb{Z}$. We can see, without difficulty, that the limit:
$$\lim_{m \rightarrow \infty} \frac{\lfloor c \sqrt{m} \rfloor}{\sqrt{m}} = c$$
By direct application of Theorem 14, we can notice that the CSS code $CSS(C_1,C_2)$ defined by:
$$C_1(m) = RM \left( \left \lfloor \frac{m-1}{2} \right \rfloor - \lfloor c \sqrt{m} \rfloor, m \right)$$
$$C_2(m) = RM(r_2(m),m)$$
is a CSS-T code with non-vanishing asymptotic rate $\Phi(-2c)$. Thus, by this construction, it is possible to generate a CSS-T code with nonvanishing asymptotic rate $\Phi(-2c)$, for any $c \in \N$.{\\}
\end{example}
\begin{example}
Let us consider any linear function, $t(m) = \lfloor a \cdot m \rfloor$, with $0 < a < \frac{1}{8}$. We can see that the CSS code $CSS(C_1,C_2)$ defined by:
$$C_1 \left( \left \lfloor  \frac{m-1}{2} \right \rfloor - \lfloor a m  \rfloor ,m \right)$$
$$C_2 \left( r_2(m),m \right)$$
where:
$$r_2(m) = 
\begin{cases}
2 t(m) + 1  & m \, \, \, \, even \\
2 t(m) & m \, \, \, \, odd 
\end{cases}$$
is a CSS-T code with vanishing rate, for any appropriately chosen $a$.  
\end{example}

\section{Conclusions}
From Theorem 14, we have derived not only a meaningful classification for CSS-T codes of this type but, furthermore, as a corollary, a mechanism for generating CSS-T codes that have their properties and parameters completely determined. Moreover, we give a useful classification of which type of Reed-Muller codes would yield optimal results via this construction. We want to thank Robert Calderbank and Narayanan Rengaswamy for their fruitful comments about this research. 

\bibliographystyle{AIMS}
\bibliography{references.bib}

\providecommand{\href}[2]{#2}
\providecommand{\arxiv}[1]{\href{http://arxiv.org/abs/#1}{arXiv:#1}}
\providecommand{\url}[1]{\texttt{#1}}
\providecommand{\urlprefix}{URL }
\begin{thebibliography}{10}

\bibitem{eczoo_css}
\newblock {C}alderbank-{S}hor-{S}teane ({CSS}) stabilizer code,
\newblock in \emph{The Error Correction Zoo} (eds. V.~V. Albert and P.~Faist),
  2022,
\newblock \urlprefix\url{https://errorcorrectionzoo.org/c/css}.

\bibitem{ourpaper}
\newblock E.~Andrade, J.~Bolkema, T.~Dexter, H.~Eggers, V.~Luongo,
  F.~Manganiello and L.~Szramowski,
\newblock Css-t codes from reed muller codes for quantum fault tolerance, 2023,
\newblock \urlprefix\url{https://arxiv.org/abs/2305.06423}.

\bibitem{berardini_structure_2024}
\newblock E.~Berardini, A.~Caminata and A.~Ravagnani,
\newblock Structure of {CSS} and {CSS}-{T} quantum codes,
\newblock \emph{Designs, Codes and Cryptography}, \textbf{92} (2024),
  2801--2823,
\newblock \urlprefix\url{https://doi.org/10.1007/s10623-024-01415-9}.

\bibitem{MSD}
\newblock S.~Bravyi and J.~Haah,
\newblock Magic-state distillation with low overhead,
\newblock \emph{Phys. Rev. A}, \textbf{86} (2012), 052329,
\newblock \urlprefix\url{https://link.aps.org/doi/10.1103/PhysRevA.86.052329}.

\bibitem{eamer-alberto}
\newblock E.~Byrne and A.~Ravagnani,
\newblock Partition-balanced families of codes and asymptotic enumeration in
  coding theory,
\newblock \emph{Journal of Combinatorial Theory, Series A}, \textbf{171}
  (2020), 105169,
\newblock
  \urlprefix\url{https://www.sciencedirect.com/science/article/pii/S0097316519301505}.

\bibitem{CSofCSS}
\newblock A.~R. Calderbank and P.~W. Shor,
\newblock Good quantum error-correcting codes exist,
\newblock \emph{Phys. Rev. A}, \textbf{54} (1996), 1098--1105,
\newblock \urlprefix\url{https://link.aps.org/doi/10.1103/PhysRevA.54.1098}.

\bibitem{campsmoreno2024quantumcsstcodessparse}
\newblock E.~Camps-Moreno, H.~H. López, G.~L. Matthews and E.~McMillon,
\newblock Toward quantum css-t codes from sparse matrices, 2024,
\newblock \urlprefix\url{https://arxiv.org/abs/2406.00425}.

\bibitem{camps-moreno_algebraic_2024}
\newblock E.~Camps-Moreno, H.~H. López, G.~L. Matthews, D.~Ruano, R.~San-José
  and I.~Soprunov,
\newblock An algebraic characterization of binary {CSS}-{T} codes and cyclic
  {CSS}-{T} codes for quantum fault tolerance,
\newblock \emph{Quantum Information Processing}, \textbf{23} (2024), 230,
\newblock \urlprefix\url{https://doi.org/10.1007/s11128-024-04427-5}.

\bibitem{csstallerton}
\newblock E.~Camps-Moreno, H.~H. López, G.~L. Matthews, D.~Ruano,
  R.~San–José and I.~Soprunov,
\newblock Binary triorthogonal and css-t codes for quantum error correction,
\newblock in \emph{2024 60th Annual Allerton Conference on Communication,
  Control, and Computing}, 2024,
\newblock 01--06.

\bibitem{gottesman1997stabilizer}
\newblock D.~E. Gottesman,
\newblock \emph{Stabilizer Codes and Quantum Error Correction},
\newblock PhD thesis, California Institute of Technology, 1997.

\bibitem{huffman_pless_2003}
\newblock W.~C. Huffman and V.~Pless,
\newblock \emph{Fundamentals of Error-Correcting Codes},
\newblock Cambridge University Press, 2003.

\bibitem{pfister-friends}
\newblock S.~Kudekar, S.~Kumar, M.~Mondelli, H.~D. Pfister, E.~Şaşoǧlu and
  R.~L. Urbanke,
\newblock Reed–{M}uller codes achieve capacity on erasure channels,
\newblock \emph{IEEE Transactions on Information Theory}, \textbf{63} (2017),
  4298--4316.

\bibitem{lidar_brun_2013}
\newblock D.~A. Lidar and T.~A. Brun,
\newblock \emph{Quantum error correction},
\newblock Cambridge University Press, 2013.

\bibitem{CSS-T-isit}
\newblock N.~Rengaswamy, R.~Calderbank, M.~Newman and H.~D. Pfister,
\newblock Classical coding problem from transversal {T} gates,
\newblock in \emph{2020 IEEE International Symposium on Information Theory
  (ISIT)}, 2020,
\newblock 1891--1896.

\bibitem{CSS-T-it}
\newblock N.~Rengaswamy, R.~Calderbank, M.~Newman and H.~D. Pfister,
\newblock On optimality of {CSS} codes for transversal {T},
\newblock \emph{IEEE Journal on Selected Areas in Information Theory},
  \textbf{1} (2020), 499--514.

\bibitem{SofCSS}
\newblock A.~M. Steane,
\newblock Simple quantum error-correcting codes,
\newblock \emph{Phys. Rev. A}, \textbf{54} (1996), 4741--4751,
\newblock \urlprefix\url{https://link.aps.org/doi/10.1103/PhysRevA.54.4741}.

\bibitem{st99}
\newblock A.~Steane,
\newblock Quantum reed-muller codes,
\newblock \emph{IEEE Transactions on Information Theory}, \textbf{45} (1999),
  1701--1703.

\bibitem{zh97}
\newblock L.~Zhang and I.~G. Fuss,
\newblock Quantum {R}eed-{M}uller codes, 1997.

\end{thebibliography}
\end{document}